\documentclass[12pt]{article}

\usepackage{amsmath,amsthm,amsfonts,amssymb,latexsym,verbatim}
\usepackage{fullpage}

\newtheorem{definition}{Definition}
\newtheorem{conjecture}{Conjecture}
\newtheorem{corollary}{Corollary}
\newtheorem{proposition}{Proposition}

\newtheorem{theorem}{Theorem}

\newcommand{\GG}{G\"{o}del's Universe}
\newcommand{\MM}{Min\-kow\-ski spacetime}
\newcommand{\SYP}{\otimes_s}
\newcommand{\AFF}{\mathcal{A}}
\newcommand{\HOM}{\mathcal{H}}
\newcommand{\KIL}{\mathcal{K}}
\newcommand{\CON}{\mathcal{C}}

\begin{document}

\title{Tensor Generalizations of\\ Affine Symmetry Vectors}
\author{Samuel A. Cook%
\footnote{Email: \texttt{scook@math.oregonstate.edu}}
~and
Tevian Dray%
\footnote{Email: \texttt{tevian@math.oregonstate.edu}}
\\
Department of Mathematics\\
Oregon State University\\
Corvallis, OR  97331
}
\date{\today}

\maketitle

\begin{abstract}
A definition is suggested for \emph{affine symmetry tensors}, which generalize
the notion of affine vectors in the same way that (conformal) Killing tensors
generalize (conformal) Killing vectors.  An identity for these tensors is
proved, which gives the second derivative of the tensor in terms of the
curvature tensor, generalizing a well-known identity for affine
vectors.  Additionally, the definition leads to a good definition of
\emph{homothetic tensors}.  The inclusion relations between these types of
tensors are exhibited.  The relationship between affine symmetry tensors and
solutions to the equation of geodesic deviation is clarified, again extending
known results about Killing tensors.
\end{abstract}


\section{Introduction}
 
An \emph{affine vector field} on a spacetime is a vector field whose associated
diffeomorphism maps geodesics to geodesics and preserves the affine parameter.
Affine vectors are also \emph{Jacobi fields}; that is, they are solutions to
the equation of geodesic deviation.  The most well-known affine vectors are
\emph{Killing vectors}, which give directions along which the spacetime metric
is invariant.  Less well-known are \emph{homothetic vectors}, which give
directions along which the metric is scaled by a constant.

In this paper, we propose definitions for \emph{affine tensor fields} and
\emph{homothetic tensor fields}, and show that their properties are direct
generalizations of the corresponding properties of affine vector fields and
homothetic vector fields, respectively.

We introduce affine tensors in Section~\ref{ATDefs}, and in
Section~\ref{ATMMGG} we review what is known about the affine algebras in
\MM\ and \GG, and show that there are nontrivial examples of affine tensors of
valence~2 in both spacetimes.

In Section~\ref{ATID}, we recall an important identity satisfied by affine
vectors, which relates the second derivative of the vector to the vector
itself by way of the curvature tensor.  We then give a generalization of this
identity to affine tensors of higher valence, which also generalizes an
identity given by Collinson~\cite{cC71} for Killing tensors of valence~2.

In Section~\ref{HTDefs}, we propose a definition of \emph{homothetic tensor},
and show that the inclusion relationships between Killing, homothetic, and
affine vectors can all be generalized to the corresponding tensors of higher
valence.  In section~\ref{HTMMGG} we exhibit nontrivial examples of homothetic
tensors of valence~2 in \MM\ and \GG.

In Section~\ref{EGD}, we show that by repeated geodesic contraction, Jacobi
fields can be systematically constructed from affine tensor fields.  This
extends and clarifies the work of Caviglia, \emph{et al.}~\cite{gCcZ82.2}.

We conclude the paper in Section~\ref{Discussion} with a brief discussion of
the Schouten bracket, and how it might be used to impose a graded Lie
structure on the set of all affine tensors on a spacetime.


\section{Affine Tensors}
\label{Affine}

\subsection{Definitions}
\label{ATDefs}

We begin with the definitions of Killing, homothetic, and affine vector
fields.

\begin{definition}
\label{Vecs}
\hfill\break\\[-10pt]
\null\hspace{24pt}
\begin{minipage}{5in}
   i) $X$ is a \emph{Killing vector field} iff $\nabla_{(a}X_{b)}=0$;\\
\smallskip
   ii) $X$ is a \emph{homothetic vector field} iff
   $\nabla_{(a}X_{b)}=cg_{ab}$, with $c$ constant;\\
\smallskip
   iii) $X$ is an \emph{affine vector field} iff
   $\nabla_a\nabla_{(b}X_{c)}=0$;
\end{minipage}\\
where in each case parentheses denote symmetrization of the enclosed indices,
and where $\nabla$ is the metric connection.
\end{definition}

The set of all affine vector fields on a spacetime ($M$,$g_{ab}$), here
denoted $\AFF(M)$, forms a Lie algebra, called the affine algebra, under the
usual Lie bracket.  The sets of Killing vector fields, $\KIL(M)$, and of
homothetic vector fields, $\HOM(M)$, are Lie subalgebras of $\AFF(M)$.
Moreover, $\KIL(M)\subset\HOM(M)\subset\AFF(M)$.

Killing tensors are a generalization of Killing vectors.  As with Killing
vectors, Killing tensors generate first integrals of the geodesic equations by
repeated contraction with geodesic tangents~\cite{lpE97}.  Interest in Killing
tensors surged in 1970 after Walker and Penrose~\cite{mWrP70} exhibited the
Killing tensor which gives the fourth constant of the motion in the Kerr
spacetime~\cite{bC1968}.  Killing tensors are defined by generalizing
definition~\ref{Vecs}($i$).

\begin{definition}
A symmetric tensor field $X_{a_1\cdots a_n}=X_{(a_1\cdots a_n)}$ is a
\emph{Killing tensor} if and only if
\begin{equation}
   \nabla_{(c}X_{a_1\cdots a_n)}=0.
\end{equation}
\end{definition}

We propose a similar generalization from affine vectors to \emph{affine
symmetry tensors}, or, more loosely, just \emph{affine tensors}.
\footnote{This terminology, though natural, has a drawback in that the term
\emph{affine tensor} is used elsewhere, notably in mechanics.}
This is done in a straightforward manner, by generalizing
definition~\ref{Vecs}($iii$).

\begin{definition}
\label{AffTenDef}
A symmetric tensor field $X_{a_1\cdots a_n}=X_{(a_1\cdots a_n)}$ is an
\emph{affine tensor} if and only if
\begin{equation}
   \nabla_b\nabla_{(c}X_{a_1...a_n)}=0.
\label{AffTen1}
\end{equation}
\end{definition}

We remark that definition~(\ref{AffTenDef}) implies that the
symmetrized covariant derivative of an affine tensor field is itself a
covariantly constant tensor field.  Extending the usual notation~\cite{gH04},
let $h_{a_1\cdots a_{n+1}}=\nabla_{(a_1}X_ {a_2\cdots a_{n+1})}$.  By the
Ricci identities and the fact that $X$ is totally symmetric, we find
\begin{equation}
   h^{\phantom{}}_{(a_1\cdots a_n|p|}R^p_{\phantom{p}a_{n+1})rs}=0.
\end{equation}

\subsection{Examples}
\label{ATMMGG}

As the concept of affine symmetry tensor is new, it is an interesting question
whether any such tensor fields exist.  The existence of valence~2 affine
tensors was investigated in~\cite{sC09}, and nontrivial examples were found in
\MM\ and \GG.

The affine structure of \MM\ is well known~\cite{gH04}.  There are 20 affine
vectors in \MM, which can be chosen so that ten are Killing vectors.
Additionally, there are~70 valence~2 affine tensor fields, which can be chosen
so that~50 are Killing tensors.  All of these tensors are reducible, that is,
expressible as symmetrized outer products of affine or Killing vectors,
respectively.

The line element for \GG\ can be written in the form
\begin{equation}
\label{LinElt}
   ds^2=dt^2-dx^2+\frac{1}{2}e^{2x}dy^2-dz^2+2e^xdt\,dy.
\end{equation}
In \GG, there are six independent affine vectors, which can be chosen so that
five are Killing vectors.  Additionally, there are 17 valence~2 affine tensor
fields.  These can be chosen so that 15 are Killing tensors, all of which are
reducible.  (The metric is a Killing tensor, and is indeed
reducible~\cite{sC09}.)  In the coordinates defined by (\ref{LinElt}), the
remaining two, the proper affine tensors of valence~2, can be written as
\begin{align}
   S_{ab} & =z\zeta_a\zeta_b, \textup{and}\\
   T_{ab} & =z(g_{ab}-\zeta_a\zeta_b),
\end{align}
where $\zeta_a=z,_a$ is covariantly constant in \GG, and $z\zeta_a$ is the
proper affine vector.  $T_{ab}$ is not reducible, but $S_{ab}$ clearly is,
being the symmetric product of the proper affine vector and the covariantly
constant vector in \GG.  In general, if there is a valence~$m$ proper affine
tensor and a valence~$n$ covariantly constant tensor, then the symmetric
product of the two will be a proper affine tensor of valence~$m+n$.  Thus, for
example, metrics which possess both covariantly constant and proper affine
vector fields, such as the Einstein Static Universe \cite{gH04}, will also
possess proper valence~2 affine tensor fields.

\subsection{Characterizations of Affine Tensors}
\label{ATID}

An important property of affine vector fields, which is known to be equivalent
to Definition~\ref{Vecs}($iii$), is that their second derivatives can be
related directly to the vector itself using the curvature of the spacetime.
If $X$ is an affine vector field, then
\begin{equation}
    \nabla_r\nabla_sX^a=R^a_{\phantom{a}srp}X^p.
\label{AffCurv}
\end{equation}
This is a well-known property of affine vectors, and is nicely discussed by
Hall~\cite{gH04}.  A similar expression for for higher valence tensors, 
\[
   \nabla_r\nabla_sX_{a_1\cdots a_n}=R^p_{\phantom{p}rs(a_1}X_{a_2\cdots a_n)p},
\]
fails to characterize affine tensors, since, for example, there are affine
tensor fields in \MM\ which have components that are quadratic in the standard
rectangular coordinates.  In this section, we give some identities relating
the second derivatives of affine tensor fields to the curvature of the
underlying spacetime.  In Theorem~\ref{ATID2} below, we present an identity
involving the unsymmetrized second derivative which can be used to
characterize affine tensor fields, and which generalizes~(\ref{AffCurv}).

\begin{theorem}
If $X$ is an affine tensor field, then
\begin{align}
\label{ATID1}
   \nabla_r\nabla_sX_{a_1\cdots a_n}-\frac{n(n-1)}{2}
	& \nabla_{(a_1}\nabla_{a_2}X_{a_3\cdots  a_n)rs}\\\notag
   &= nR^p_{\phantom{p}rs(a_1}X^{\phantom{p}}_{a_2\cdots a_n)p}
	- n(n-1)R^p_{\phantom{p}(a_1a_2|(r}X^{\phantom{p}}_{s)|a_3\cdots a_n)p},
\end{align}
where the vertical bars excuse the enclosed indices from symmetrization.
\label{thm1}
\end{theorem}
\begin{proof}
This is proved by repeatedly employing the definition of affine
tensor fields, the Ricci identities, and collecting terms into the proper
symmetrizations, for which it is useful to recall the identity
\begin{equation}
   T_{(a_1a_2\cdots a_n)}
   =\tfrac{1}{n}(T_{a_1(a_2\cdots a_n)}+\cdots +T_{a_n(a_1\cdots a_{n-1})}),
\end{equation}
valid for any tensor $T_{a_1a_2\cdots a_n}$.

Applying in succession the definition of affine tensor, the definition
of the curvature tensor, and again the definition of affine tensor
gives
\begin{align}
   \nabla_r\nabla_sX_{a_1\cdots a_n}=&\phantom{-}
   \nabla_{a_1}\nabla_sX_{ra_2\cdots a_n}+
   \nabla_{a_1}\nabla_{a_2}X_{sra_3\cdots a_n}+\cdots+\\\notag
     &  \nabla_{a_2}\nabla_{a_1}X_{a_3\cdots a_nrs}+\cdots+
     \nabla_{a_n}\nabla_sX_{a_1\cdots a_{n-1}r}-\\\notag
     & \{R^p_{\phantom{p}sa_1r}X_{pa_2\cdots
       a_n}+R^p_{\phantom{p}a_na_1r}X_{sa_2\cdots a_{n-1}p}+
     \cdots +R^p_{\phantom{p}sa_nr}X_{a_1\cdots a_{n-1}p}\}
\end{align}
Rearranging terms, applying the Ricci identities, the affine condition, the
Ricci identities once more, and organizing the curvature terms, gives
\begin{align}
    \nabla_r\nabla_sX_{a_1\cdots a_n}= & - \nabla_r\nabla_sX_{a_1\cdots a_n}
    +n(n-1)\nabla_{(a_1}\nabla_{a_2}X_{a_3\cdots a_n)(rs)}\\\notag    
     & +2nR^p_{\phantom{p}rs(a_1}X^{\phantom{p}}_{a_2\cdots a_n)p}-
      2n(n-1)R^p_{\phantom{p}(a_1a_2|(r}X^{\phantom{p}}_{s)|a_3\cdots a_n)p},
\end{align}
and the result follows.  Further details may be found in~\cite{sC09}.
\end{proof}

Theorem~\ref{thm1} generalizes a result of Col\-linson~\cite{cC71}, who
developed the following identity for Killing tensors.

\begin{corollary}
Let $T$ be a valence $2$ Killing tensor field.  Then
\begin{equation}
   \nabla_r\nabla_sT_{ab} - \nabla_b\nabla_aT_{rs}
	= 2R^p_{\phantom{p}rs(a}T^{\phantom{p}}_{b)p}
	  - 2R^p_{\phantom{p}ba(s}T^{\phantom{p}}_{r)p}
\end{equation}
\end{corollary}
\begin{proof}
This follows by letting $n=2$ in Theorem~\ref{thm1}, and expanding the
symmetrizations.
\end{proof}

It should be remarked that even though Collinson proved this only for
Killing tensors, it is in fact true for all affine tensors.

\begin{corollary}
If $X$ is an affine tensor field in \MM, then 
\begin{equation}
\label{NoAff1}
   \partial_r\partial_sX_{a_1\cdots a_n}
	= \frac{n(n-1)}{2}\partial_{(a_1}\partial_{a_2}X_{a_3\cdots a_n)rs}.
\end{equation}
\end{corollary}

The converse of Theorem~\ref{thm1} is false.  For example, in \MM, in standard
rectangular coordinates, the tensor $T_{ab}=xyX_{(a}Y_{b)}$ with
$X=\partial_x$, $Y=\partial_y$ satisfies~(\ref{NoAff1}), but is not an affine
tensor.
\footnote{The claimed proof of the converse of Theorem~\ref{thm1}
in~\cite{sC09} is therefore incorrect.}

The following theorem gives an identity for affine tensor fields which
generalizes~(\ref{AffCurv}), and which does fully characterize affine
tensor fields.

\begin{theorem}
A tensor field $X_{a_1\cdots a_n}=X_{(a_1\cdots a_n)}$ is affine if and only
if
\begin{align}
\label{AffID2}
    \nabla_r\nabla_sX_{a_1\cdots a_n} & = \frac{2n}{n+1}R^p_{\phantom{p}rs(a_1}X^
{\phantom{}}_{a_2\cdots a_n)p}\\\notag
    & + \frac{n(n-1)}{n+1}(\nabla_{(a_1}\nabla_{a_2}X_{a_3\cdots a_n)sr}-\nabla_s
\nabla_{(a_1}X_{a_2\cdots a_n)r})\\\notag
    & +\frac{n(n-1)}{n+1}(R^p_{\phantom{p}(a_1|sr|}X^{\phantom{}}_{a_2\cdots a_n)
p}-2R^p_{\phantom{p}(a_1a_2|(r}X^{\phantom{}}_{s)|a_3\cdots a_n)p})\notag
\end{align}
\label{ATID2}
\end{theorem}

\begin{proof}
Suppose $X$ is an affine tensor.  By the Ricci identity and the symmetry 
of $X$,
\begin{equation}
   \nabla_r\nabla_sX_{a_1\cdots a_n}
   = \nabla_s\nabla_r X_{a_1\cdots a_n}
     + nR^p_{\phantom{p}(a_1|sr|}X^{\phantom{}}_{a_2\cdots a_n)p}
\end{equation}
By the affine condition,
\begin{equation}
\label{AffIDPart2}
   \nabla_r\nabla_sX_{a_1\cdots a_n}
   = -n\nabla_s\nabla_{(a_1}X_{a_2\cdots a_n)r}
     + nR^p_{\phantom{p}(a_1|sr|}X^{\phantom{}}_{a_2\cdots a_n)p}
\end{equation}
Equation~(\ref{AffID2}) follows from adding (\ref{ATID1}) with
$\tfrac{n-1}{2}$ times equation~(\ref{AffIDPart2}).  On the other hand, if a
symmetric tensor field $X$ satisfies~(\ref{AffID2}), then by symmetrizing on
all indices except $r$, it follows that it must be affine.
\end{proof}

We may recover the standard result~(\ref{AffCurv}) for affine vector fields
from Theorem~\ref{ATID2} by letting $n=1$ (and rearranging indices).
\begin{corollary}
$X$ is an affine vector field if and only if
\begin{equation}
   \nabla_r\nabla_sX_a=R^p_{\phantom{p}rsa}X_p.
\end{equation}
\end{corollary}

\section{Homothetic Tensors}
\label{HomT}

\subsection{Definitions}
\label{HTDefs}

As mentioned above, there are several types of vectors which are of interest
in the study of spacetime, including Killing vectors, homothetic vectors, and
affine vectors.  Another important class of vector is that of \emph{conformal
Killing vectors}.  A conformal Killing vector gives a direction along which
the metric is rescaled.  Conformal Killing tensors are defined from conformal
Killing vectors in a manner analogous to the way Killing tensors were defined.
As with conformal Killing vectors, conformal Killing tensors generate
constants of the motion along null geodesics.

\begin{definition}
\label{CKs}
\hfill\break\\[-10pt]
\null\hspace{24pt}
\begin{minipage}{5in}
   i) $X$ \textup{is a conformal Killing vector iff }
   $\nabla_{(a}X_{b)}=\phi g_{ab}$;\\
\smallskip
   ii) $X$ \textup{is a conformal Killing tensor iff }
   $\nabla_{(c}X_{a_1\cdots a_n)}=g_{(ca_1}\phi_{a_2\cdots a_n)}$. 
\end{minipage}\\
where $\phi$ is a scalar field (not necessarily constant), and where the
tensor fields $\phi_{a_1\cdots a_{n-1}}$ and $X_{a_1\cdots a_n}$ are totally 
symmetric.
\end{definition}

The set of conformal Killing vector fields on a spacetime, denoted $\CON(M)$,
forms a Lie algebra, called the conformal algebra.  Clearly, $\KIL(M)$ and
$\HOM (M)$ form subalgebras of $\CON(M)$, and
$\KIL(M)\subset\HOM(M)\subset\CON(M) $.  However, in general
$\AFF(M)\neq\CON(M)$.  In fact, as Hall has pointed out~\cite {gH04},
$\HOM(M)=\AFF(M)\cap\CON(M)$.  By making a suitable definition of a valence
$n$ homothetic tensor, we will be able to generalize all of these
relationships.

\begin{definition}
A symmetric tensor $X_{a_1a_2\cdots a_n}=X_{(a_1a_2\cdots a_n)}$ is a
homothetic tensor if and only if
\begin{equation}
   \nabla_{(c}X_{a_1...a_n)}=g_{(ca_1}\lambda_{a_2...a_n)},
\label{HomTen1}
\end{equation}
for some covariantly constant tensor $\lambda_{a_1\cdots a_{n-1}}=\lambda_
{(a_1\cdots a_{n-1})}$.
\end{definition}

We note that Prince~\cite{gP83} has offered a definition of \emph{homothetic
Killing tensors} which should not be confused with the definition just given.
Prince is referring to tensors constructed on the evolution space, and the
homothetic factor is a Killing vector so that the homothetic nature of the
tensor only appears along geodesics.

Let $\KIL^n(M)$, $\HOM^n(M)$, $\AFF^n(M)$, and $\CON^n(M)$ respectively denote
the sets of valence $n$ Killing, homothetic, affine, and conformal Killing
tensor fields.  Leaving aside for a moment the question of whether any or all
of these sets enjoy some sort of Lie algebra structure, we can remark that
similar set relations obtain.

\begin{proposition}
\hfill\break
\null\hspace{24pt}
\begin{minipage}{5in}
   i) $\KIL^n(M) \subset \HOM^n(M) \subset \CON^n(M)$;\\\smallskip
   ii) $\KIL^n(M) \subset \HOM^n(M) \subset \AFF^n(M)$;\\\smallskip
   iii) $\HOM^n(M)=\CON^n(M)\cap \AFF^n(M)$.
\end{minipage}
\end{proposition}

\begin{proof}
The proofs of (i) and (ii) are straightforward, and together they show that
$\HOM^n(M)\subset \CON^n(M)\cap \AFF^n(M)$.

To show that $\CON^n(M)\cap \AFF^n(M)\subset \HOM^n(M)$, choose $X\in
\CON^n(M)\cap \AFF^n(M)$.  Then there is a valence~$n-1$ symmetric
tensor $\lambda$ such that
\begin{equation}
   \nabla_c\nabla_{(a_1}X_{a_2\cdots
     a_{n+1})}=\nabla_c(g_{(a_1a_2}\lambda_{a_3\cdots a_{n+1})})=0
\end{equation}
This implies that
\begin{equation}
   g_{a_1a_2}\nabla_c\lambda_{a_3\cdots a_{n+1}}+\cdots
   +g_{a_{n+1}a_n}\nabla_c\lambda_{a_1\cdots a_{n-1}}=0
\label{LambContra}
\end{equation}

To prove the claim, we must show that $\nabla_c\lambda_{a_1\cdots a_{n-1}}=0$.
To do this, first show that all possible contractions of $\lambda$ are
covariantly constant.  This is achieved by first transvecting
(\ref{LambContra}) with $g^{a_1a_2}g^{a_3a_4}\cdots g^{a_{n-1}a_n}$ if $n$ is
even, or $g^{a_1a_2}g^{a_3a_4}\cdots g^{a_na_{n+1}}$ if $n$ is odd, showing
that the maximally contracted form of $\lambda$ is covariantly constant.  This
is repeated, each time transvecting (\ref{LambContra}) with one less factor of
$g$, until it is shown that all contractions of $\lambda$ are covariantly
constant.  Finally, transvecting (\ref{LambContra}) with $g^{a_na_{n+1}}$
gives $\nabla_c\lambda_{a_1\cdots a_{n-1}}=0$, which proves the claim.
\end{proof}

\begin{corollary}
$\HOM(M)=\AFF(M)\cap\CON(M)$.
\end{corollary}

\subsection{Examples}
\label{HTMMGG}

There are nontrivial examples of valence 2 homothetic tensor fields in
\MM\ and \GG.  In \MM, $(\mathbb{R}^4,\eta_{ab})$ in the usual rectangular
coordinates, the four tensors $t\eta_{ab}$, $x\eta_{ab}$, $y\eta_{ab}$, and
$z\eta_{ab}$ are independent homothetic tensors.  In \GG, the sum
$X_{ab}=S_{ab}+T_{ab}=zg_{ab}$ is homothetic, as is easy to see, since
\begin{equation}
   \nabla_{(a}X_{bc)}=g_{(bc}\nabla_{a)}z=z_{,(a}g_{bc)}=\zeta_{(a}g_{bc)}.
\end{equation}

Metrics which possess a valence~$m$ covariantly constant tensor and a
valence~$n$ proper homothetic tensor will posses a reducible homothetic tensor
of valence~$m+n$ given by the symmetric product of the two. Thus, for example,
some of the plane wave metrics studied in \cite{gH04} possess a reducible
homothetic tensor of valence~2.

\section{Geodesic Deviation}
\label{EGD}

In 1982, Caviglia \emph{et al.}~\cite{gCcZ82.2} showed that by repeatedly
contracting a Killing tensor with copies of a particular geodesic tangent
vector, one obtained a \emph{Jacobi field}, that is, a solution of the
equation of geodesic deviation.  They were able to obtain a partial converse,
essentially showing that if a tensor always yields a Jacobi field under this
construction, then it is an affine tensor, although they did not employ this
notion in their work.

The property that a Killing tensor generates a Jacobi field is in fact a
property it possesses because it is an affine tensor.  The following theorem,
together with the results in~\cite{gCcZ82.2}, shows that a tensor generates a
Jacobi field through this process if and only if it is an affine tensor.

\begin{theorem}
Suppose $X_{aa_1...a_p}$ is an affine tensor, and let 
\[
   \theta_a=X_{aa_1...a_p}V^{a_1}\cdots V^{a_p},
\]
for any geodesic tangent $V$.  Then $\theta^a$ satisfies the \emph{equation of
geodesic deviation}
\[
V^bV^c\nabla_b\nabla_c\theta_a
	+ R_{ab\phantom{c}d}^{\phantom{ab}c}V^b\theta_cV^d=0.
\]
\label{thm2}
\end{theorem}

\begin{proof}
As with Theorem~\ref{thm1}, the result follows by repeatedly employing the
definitions of affine tensors and the curvature tensor, and appropriate index
symmetrization, in combination with the geodesic equations.  In this case, the
repeated contraction by a geodesic tangent induces the required
symmetrization.  For the complete proof, see~\cite{sC09}
\end{proof}

\section{Discussion}
\label{Discussion}

Geroch~\cite{rG70}, Sommers~\cite{pS73}, and Thompson~\cite{gT86} have
considered a Lie algebra structure on the set $\cup_n\KIL^n(M)$, using the the
Schouten bracket.  The Schouten bracket is a generalization of the standard
Lie bracket, defined on pairs of symmetric tensors $S$ and $T$ by
\begin{equation}
   [S,T]^{a_1\cdots  a_{m+n-1}}=\\
    mS^{r(a_1\cdots a_{m-1}}\nabla_rT^{a_m\cdots a_{m+n-1})}
	- nT^{r(a_1\cdots a_{n-1}}\nabla_rS^{a_n\cdots a_{m+n-1})}
\label{Schouten}
\end{equation}
The Schouten bracket is antisymmetric, linear in each slot, satisfies the
Jacobi identity, and reduces to the ordinary Lie bracket if $S$ and $T$ are
vector fields.  The Schouten bracket also satisfies
\begin{equation}
   [S,T\SYP V]=[S,T]\SYP V + [S,V]\SYP T
\label{SYMPRODBRACK}
\end{equation}
for any symmetric tensors $S$, $T$, $V$, where $T\SYP V$ represents the
symmetrized outer product of $T$ and $V$.

It is well known that the sets of Killing and conformal Killing tensors are
closed under the operation of taking symmetrized outer products.  Note,
however, that neither the set of all homothetic tensors, nor the set of all
affine tensors is closed under this operation.

Geroch~\cite{rG70} has pointed out that Killing tensors may be defined as
those symmetric tensors in a spacetime which commute with the metric $g$ under
this bracket; $K$ is a Killing tensor if and only if $[K,g]=0$.  It follows
from this and the Jacobi identity that $\cup_n\KIL^n(M)$ is closed
under~\eqref{Schouten}, and so forms a graded Lie algebra with respect to the
Schouten bracket.

Geroch remarks, additionally, that similar comments apply to the set of
conformal Killing tensors, where now the condition reads $[K,g]=g\tilde{K}$
for some symmetric tensor $\tilde{K}$ whose valence is one less than
that of~$K$.  Using this condition, the Jacobi identity,
and~\eqref{SYMPRODBRACK}, it can be seen that $\cup_n\CON^n(M)$ is closed
under~\eqref{Schouten}, and so also forms a graded Lie algebra with respect to
the Schouten bracket.

It is clear how to extend this structure to the set $\cup_n\HOM^n(M)$ of
homothetic tensors in a spacetime.  With regards to the set $\cup_n\AFF^n(M)$
of affine tensors, we pose the following conjecture.
\begin{conjecture}
The set $\cup_n\AFF^n(M)$ of all affine tensors in a spacetime is closed under
the Schouten bracket, and can therefore be given a graded Lie algebra
structure.
\end{conjecture}

Recall that, in terms of the standard Lie derivative, an affine vector $X$ is
defined by $\mathcal{L}_X\nabla=0$~\cite{gH04}.  A natural extension of this
definition to affine tensors, which would also be an extension of the above
considerations, would be something like the following: Under the Schouten
bracket, a symmetric tensor $X$ is affine if and only if $[X,\nabla]=0$.
Analysis of affine tensor fields along these lines may illuminate the
geometric role these tensor fields play.  This investigation is ongoing.





\end{document}